\documentclass[envcount]{article}
\usepackage{thesisstyle}
\newcommand{\hajos}{\hyperlink{hajos}{Hajos graph}}

\usepackage{graphicx}

\begin{document}

\title{A note on coloring (even-hole,cap)-free graphs}
\author{
Shenwei Huang\\ School of Computing Science\\
Simon Fraser University, Burnaby B.C., V5A 1S6, Canada\\ 
\texttt{shenweih@sfu.ca}\\ 
\\
Murilo V. G. da Silva\thanks{Partially supported by CNPq.}\\ Universidade Tecn\'ologica Federal do Parana, Curitiba, Brazil\\
\texttt{murilo@utfpr.edu.br}
}
\date{}
\maketitle

\begin{abstract}
A {\em hole} is a  chordless cycle of length at least four.
A hole is {\em even}  (resp. {\em odd}) if it contains an even (resp. odd) number of vertices.
A \emph{cap} is a graph  induced by a hole with an additional vertex that is adjacent to exactly
two adjacent vertices on the hole.
In this note, we use a decomposition theorem by Conforti et al. (1999)
to show that if a graph $G$ does not contain any even hole or cap as an induced subgraph,
then $\chi(G)\le \lfloor\frac{3}{2}\omega(G)\rfloor$, where
$\chi(G)$ and $\omega(G)$ are the chromatic number and the clique number of $G$, respectively.
This bound is attained by odd holes and the Hajos graph.
The proof leads to a polynomial-time $3/2$-approximation algorithm for coloring (even-hole,cap)-free graphs.
\end{abstract}

\section{Introduction}

All graphs in this paper are finite, simple and undirected.
We say that a graph $G$ {\em contains} a graph $F$, if $F$ is
isomorphic to an induced subgraph of $G$.
A graph $G$ is
{\em $F$-free} if it does not contain $F$.
Let ${\cal F}$ be a (possibly infinite) family of graphs.
A graph $G$ is
{\em ${\cal F}$-free} if it is $F$-free, for every $F \in {\cal F}$.

A {\em hole} is a
chordless cycle of length at least four.
A hole is {\em even}
(resp. {\em odd}) if it contains an even (resp. odd) number of nodes.
A hole of length $n$ is also called an {\em $n$-hole}.
We denote by a $n$-hole by $C_n$.
Let $G=(V(G),E(G))$ be a graph.
The two vertices $u,v\in V(G)$ are \emph{adjacent} (or \emph{neighbors}), respectively,
\emph{non-adjacent} (or \emph{non-neighbors}) if $\{u,v\}\in E(G)$, respectively, $\{u,v\}\notin E(G)$.
The \emph{open neighborhood} of a vertex $v$, denoted by $N_G(v)$, is the set of neighbors of $v$.
The \emph{closed neighborhood} of $v$ is $N_G[v]=N(v)\cup \{v\}$.
For a set $X\subseteq V(G)$, let $N_G(X)=\bigcup_{v\in X}N_G(v)\setminus X$
and $N_G[X]=N_G(X)\cup X$.
The \emph{degree} of $v$, denoted by $d_G(v)$, is equal to $|N_G(v)|$.
We shall omit the subscript $G$ if the context is clear.

A \emph{complete graph} is a graph so that every pair of vertices are adjacent.
We denote $K_n$ by the complete graph with $n$ vertices.
The graph $K_3$ is also called a \emph{triangle}.
A \emph{clique} is a vertex subset that induces a complete graph.
The size of a largest clique in $G$, denoted by $\omega(G)$, is the \emph{clique number} of $G$.
A vertex subset $S$ of $V$ is a  \emph{cutset} if $G-S$ has more connected components
than $G$. If a cutset $S$ is also a clique, it is called a \emph{clique cutset}.
A graph with no clique cutsets is called an \emph{atom}.
We say that a vertex 
\emph{universal}
if it is adjacent to all other vertices in $G$. 
For two subsets $X,Y\subseteq V(G)$, we say that $X$ is \emph{complete} (respectively \emph{anti-complete}) to  $Y$
if every vertex in $X$ is adjacent (respectively non-adjacent) to every vertex in $ Y$.
If $X$ consists of only a single element $x$, we simply say $x$, rather than $\{x\}$, is complete (anti-complete) to $Y$.


A (\emph{proper}) \emph{$k$-coloring} of a graph $G=(V,E)$ is a mapping $\phi:V\rightarrow \{1,2,\ldots,k\}$
such that $\phi(u)\neq \phi(v)$ whenever $uv\in E$. The value $\phi(u)$ is usually referred to
as the {\em color} of $u$ under $\phi$.
We say that $G$ is {\em $k$-colorable} if $G$ admits a $k$-coloring.
The \emph{chromatic number} of $G$, denoted by $\chi(G)$,
is the smallest positive integer $k$ such that $G$ is $k$-colorable.
We use \textsc{Chromatic Number}
and \textsc{$k$-Colorability} to denote the problem
of finding the chromatic number and deciding if a given
graph is $k$-colorable, respectively.
Throughout the paper, we use $n$ and $m$ to denote the number
of vertices and edges in $G$, respectively.

\subsection{Even-hole-free graphs}
A graph is \emph{even-hole-free} if it is $\{C_4,C_6,\ldots\}$-free.
Even-hole-free graphs generalize \emph{chordal} graphs, i.e.,
those graphs that are hole-free. 
The structure of even-hole-free graphs was first studied by
Conforti, Cornu\'ejols, Kapoor and Vu\v{s}kovi\'c in \bfcite{cckv-ehf1,cckv-ehf2}.
They focused on showing that even-hole-free
graphs can be recognized in polynomial time
(a problem that at that time
was not even known to be in NP), and their primary motivation
was to develop techniques that can then be used in the study of perfect graphs.
To state their result, we first give some definitions.

A node set $S \subseteq V(G)$ is a {\em $k$-star cutset} of $G$ if $S$
is a cutset and $S$ contains a clique of size $k$ so that every vertex in $S\setminus C$ has a neighbor in $C$.
A $1$-star, $2$-star and $3$-star are referred to as
\emph{star}, \emph{double star} and \emph{triple star},
respectively. Moreover, $S$ is said to be a \emph{full $k$-star} if $S=N[C]$.

A graph $G$ has a {\em $2$-join} $V_1|V_2$, with special sets $(A_1,A_2,B_1,B_2)$,
if the nodes of $G$ can be partitioned into sets $V_1$ and $V_2$ so that the following hold.
\begin{enumerate}[label=\itshape (J\arabic*)]
\item For each $1\le i\le 2$,  $A_i$ and $B_i$ are non-empty and disjoint with $A_i \cup B_i \subseteq V_i$.
\item $A_1$ and $B_1$ are complete to $A_2$ and $B_2$, respectively,  and  these are the only edges between $V_1$ and $V_2$.
\item For each $1\le i\le 2$, the graph $G[V_i]$ induced by $V_i$ contains a path with one end in $A_i$ and the other in $B_i$
but $G[V_i]$ is not a chordless path.
\end{enumerate}

In $2002$, Conforti, Cornu\'ejols, Kapoor and Vu\v{s}kovi\'c \bfcite{cckv-ehf1,cckv-ehf2}
obtained the first decomposition theorem for even-hole-free graphs that uses $2$-joins and star, double star and triple star cutsets.
The decomposition was then led to the first polynomial time recognition algorithm for even-hole-free graphs.
Since the main motivation was to show the existence of such an algorithm, they did not intend to optimize
the running time which is $O(n^{40})$.
Soon after, Chudnovsky, Kawarabayashi and Seymour \bfcite{cks}
developed a $O(n^{31})$ recognition algorithm. Their algorithm is not based
on decomposition theorems but on directly testing for even holes  after a certain step
called cleaning is performed.
Later on,  Silva and Vu\v{s}kovi\'c obtained a new decomposition theorem which avoids double star and triple star cutsets.
\begin{thm}\bfcite{SV13}\label{thm:decom ehf}
Every connected even-hole-free graph is either basic or admits a star cutset or a $2$-join.
\end{thm}

Here the description of `basic' graphs are somewhat technical and we refer to
\bfcite{SV13} for formal definitions.
Taking advantage of this strengthened decomposition,
Silva and Vu\v{s}kovi\'c were able to obtain an $O(n^{19})$ algorithm to recognize even-hole-free graphs
which is a significant improvement over the ones from \bfcite{cks,cckv-ehf2}.
Very recently, Chang and Lu \bfcite{CL15} showed that the $O(n^{19})$ algorithm does not
take full advantage of \autoref{thm:decom ehf}. With more advanced techniques, they developed
the best known recognition algorithm so far.
\begin{thm}\bfcite{CL15}\label{thm:recog ehf}
For a graph $G$ with $n$ nodes and $m$ edges, there exists an algorithm
that runs in $O(m^3n^5)$ to recognize if $G$ is even-hole-free.
Moreover, the algorithm outputs an even hole if it exists.
\end{thm}

\subsection{$\chi$-boundedness and $\beta$-perfectness}

Note that by excluding a $4$-hole, one also excludes all antiholes (An antihole
is the complement of a hole) of length at least $6$. If we switch parity, the closer analogous class to even-hole-free graphs
is the class of perfect graphs rather than just the odd-hole-free graphs.
It was shown \bfcite{GLS84} that \textsc{Chromatic Number}
can be solved in polynomial time for perfect graphs.  In contrast, it remains open whether one can optimally color
an even-hole-free graph
(this is also the case for \textsc{$k$-Colorability}).
\begin{prob}
What is the complexity of \textsc{Chromatic Number} for even-hole-free graphs?
\end{prob}
Despite the unknown status of the complexity of determining $\chi(G)$
for even-hole-free graphs, an approximate version does exist.
In $2008$, Addario-Berry, Chudnovsky, Havet, Reed and Seymour \bfcite{achrs} settled a conjecture of Reed
by proving that every even-hole-free graph contains a {\em bisimplicial vertex} (a vertex whose set of neighbors
induces a graph that is a union of two cliques).  Since the degree of a bisimplicial vertex is at most $2\omega (G) -2$,
this has the following immediately consequence.
\begin{thm}\bfcite{achrs}\label{thm:bisimplicial}
If $G$ is an even-hole-free graph, then $\chi (G) \leq 2 \omega (G) -1$.
\end{thm}

Gy\'arf\'as \bfcite{Gy87} introduced the concept of $\chi$-bounded graphs as
a natural extension of perfect graphs. A class $\mathcal{G}$ is called
\emph{$\chi$-bounded} with \emph{$\chi$-binding function $f$} if for every
induced subgraph $G'$ of $G$ it holds that $\chi(G')\le f(\omega(G'))$.
The class of perfect graphs is a $\chi$-bounded family with identity function $f(x)=x$
being its $\chi$-binding function. Translating \autoref{thm:bisimplicial} into this language,
it says that  the class of even-hole-free graphs  belongs to the family of $\chi$-bounded graphs
with $\chi$-binding function $f(x)=2x-1$.  On the other hand, it is well-known that
finding a maximum clique in $C_4$-free graphs (hence even-hole-free graphs) can be achieved in polynomial time.
It was first observed by Farber \bfcite{Fa89} that $4$-hole-free graphs have
$O(n^2)$ maximal cliques and all of them can be listed in polynomial time.
For even-hole-free graphs, \autoref{thm:bisimplicial} implies that the neighborhood
is chordal. 
The existence of a vertex whose neighborhood induces a chordal graph
in even-hole-free graphs was first proved by Silva and Vu\v{s}kovi\'c \bfcite{SV07}.
Since it takes linear-time to find the clique number in a chordal graph, see for example \bfcite{Golu04},
this fact implies that $\omega(G)$ can be computed in $O(mn)$
for even-hole-graph $G$ with $n$ vertices and $m$ edges.
This and \autoref{thm:bisimplicial} imply:
\begin{thm}\label{thm:2-appro}
There exists an $O(mn)$ $2$-approximation algorithm
for computing the chromatic number of even-hole-free graphs.
Moreover, the algorithm outputs a (proper) coloring of $G$
that uses at most $2\omega(G)-1$ colors.
\end{thm}

Another motivation for the study of even-hole-free graphs is their
connection to $\beta$-perfect graphs introduced by Markossian,
Gasparian and Reed \bfcite{mgr}.
For a graph $G$, consider the following linear ordering on $V(G)$:
order the vertices by repeatedly removing a vertex of minimum
degree in the subgraph of vertices not yet chosen and placing it after
all the remaining vertices but before all the vertices already removed.
Coloring greedily on this order gives the upper bound
$\chi (G) \leq \beta (G)$, where
\[\beta(G)=\max\{\delta(G')+1:G' \text{ is an induced subgraph of } G\}.\]
A graph is \emph{$\beta$-perfect} if for each induced subgraph $H$ of $G$, $\chi(H)=\beta(H)$.
Clearly, $\beta(C_{2s})=3$ and $\chi(C_{2s})=2$ for any $s\ge 2$.
This means that any $\beta$-perfect graph must be even-hole-free.
The converse of the statement is not necessarily true (replacing each vertex of a $5$-hole
by a clique of size two gives a counter-example).
Nevertheless, if we forbid an additional graph in addition to even holes, it is possible
to obtain $\beta$-perfect graphs. A recent result of Kloks, M\"uller and Vu\v{s}kovi\'c
\bfcite{kmv} showed that if the additional forbidden graph is the diamond, then
this is indeed the case. A \emph{diamond} is the graph obtained from $K_4$ by  removing an edge.
\begin{thm}\bfcite{kmv}\label{thm:diamond}
Every (even-hole,diamond)-free graph is $\beta$-perfect.
\end{thm}
The $\beta$-perfectness of (diamond,even-hole)-free graphs is a consequence
of the fact that every such graph contains a \emph{simplicial extreme}, namely a vertex
that is either simplicial or of degree two, which in turn follows from a decomposition theorem
for (diamond,even-hole)-free graphs that uses $2$-joins, clique cutsets
and bisimplicial cutsets (a special type of a star cutset).
The $\beta$-perfectness of (diamond,even-hole)-free graphs
implies that $\chi(G)$ can be computed in polynomial time
by coloring greedily on the particular ordering of vertices we described above.
\begin{coro}\label{coro:coldiamond}
\textsc{Chromatic Number} can be solved in $O(n^2)$ time for (even-hole,diamond)-free graphs.
\end{coro}

\begin{proof}
Let $G$ be a (even-hole,diamond)-free graph.
By \autoref{thm:diamond}, $G$ is $\beta$-perfect.
This implies that $\chi(G)=\beta(G)$.
More accurately, let $v_1,v_2,\ldots,v_n$ be the linear ordering
obtained from the procedure we described above, i.e.,  $v_i$
is a vertex of minimum degree in $G_i=G[\{v_1,\ldots,v_i\}]$.
Then $\chi(G)\le \max\{\delta(G_i)+1:1\le i\le n\}\le \beta(G)=\chi(G)$.
This means that $\chi(G)=\max\{\delta(G_i)+1:1\le i\le n\}$.
Clearly, it takes $O(i)$ time to find $v_i$ in $G_i$ for each $i$.
Thus, finding such a linear ordering can be done in $O(n^2)$ time.
Moreover, greedily coloring $G$ on $v_1,\ldots,v_n$ can be done in $O(m+n)$
time. Therefore, the corollary holds
\end{proof}

In addition, the existence of a simplicial extreme immediately implies that
the class of (even-hole, diamond)-free graphs is a $\chi$-bounded family with
$\chi$-binding function $f(x)=x+1$.
\begin{coro}\bfcite{kmv}\label{coro:omegaplus1}
For any (even-hole, diamond)-free graph $G$, $\chi(G)\le \omega(G)+1$.
\end{coro}

\subsection{Subclasses of even-hole-free graphs}

Very recently,  efforts are made on subclasses of even-hole-free graphs
by forbidding additional graphs  besides even holes .
The result of diamond-free graphs \bfcite{kmv} already demonstrates the richness
of this approach.

A \emph{pan} is a graph induced by a hole with an additional vertex pendent
to some vertex on the hole. 
Cameron, Chaplick and Ho\`{a}ng  \bfcite{CCH15} investigated (even-hole,pan)-free graphs.
They first obtained a decomposition theorem for (even-hole,pan)-free graphs:
every such graph can be decomposed via clique cutset into (essentially) unit circular-arc graphs.
The decomposition allows them to obtain an $O(mn)$ recognition algorithm and a polynomial time
coloring algorithm. Although the class of (even-hole,pan)-free graphs is not $\beta$-perfect, it was shown
to be $\chi$-bounded with $\chi$-binding function $f(x)=\frac{3}{2}x$.

A \emph{cap} is a graph  induced by a hole with an additional vertex that is adjacent to exactly
two adjacent vertices on the hole. 
A graph is \emph{cap-free} if it does not contain any cap as an induced subgraph.
It was shown by Conforti, Gerards and Pashkovich \bfcite{CGP15}
that the problem of weighted maximum independent set can be solved in polynomial time for (even-hole,cap)-free
graphs. 
We study \textsc{Chromatic Number} for (even-hole,cap)-free graphs below.
Like the pan-free case, (even-hole,cap)-free graphs need not
to be $\beta$-perfect. We show that the class of (even-hole,cap)-free graphs is a $\chi$-bounded family
with $\chi$-binding function $f(x)=\frac{3}{2}x$. The following is the main result in this paper.

\begin{thm}\label{thm:ehbf_main}
For any (even-hole,cap)-free graph $G$,
$\chi(G)\le \lfloor\frac{3}{2}\omega(G) \rfloor$.
\end{thm}

\section{Decomposition of cap-free graphs}\label{sec:cap}

In $1999$, Conforti, Cornu\'ejols, Kapoor and Vu\v{s}kovi\'c \bfcite{CCKV99}
proved a decomposition theorem for cap-free graphs.
To state their decomposition, we first define a special kind of `cutset'.
Let $X = (V_1, A_1, V_2, A_2, K)$ be an array of disjoint sets with union $V(G)$.
We say that $X$ is an \emph{amalgam} of $G$ if the following properties hold:

\begin{itemize}
\item $A_1$ and $A_2$ are complete to each other and both are non-empty.
\item $K$ is a clique (possibly empty) and $K$ is complete to $A_1 \cup A_2$.
\item $V_1$ is anti-complete to $A_2 \cup V_2$ and $V_2$ is anti-complete to $A_1 \cup V_1$.
\item $|V_1 \cup A_1| \geq 2$ and $|V_2 \cup A_2| \geq 2$.
\end{itemize}

Note that possibly $K$ may have neighbors in $V_1 \cup V_2$.

\begin{thm}\bfcite{CCKV99}\label{thm:decom cap}
Every cap-free graph with a triangle either admits
an amalgam or a clique cutset or contains a universal vertex.
\end{thm}

Therefore, cap-free graph can be built from triangle-free graphs.
We say that two vertices $u$ and $v$ are \emph{twins} in $G$
if $N[u]=N[v]$, and that $G$ contains twin vertices if
there are vertices that are twins in $G$.
We notice in the following that if we forbid even holes in cap-free graphs,
then an amalgam of $G$ gives rise to twin vertices.

\begin{lem}\label{lem:twin}
Suppose that $G$ is an (even-hole, cap)-free graph containing no clique cutset.
If $G$ contains an amalgam $X = (V_1, A_1, V_2, A_2, K)$, then $G$ contain a
pair of twin vertices.
\end{lem}

\begin{proof}
Suppose that both $A_1$ and $A_2$ are not cliques. Then
$A_1$ (respectively $A_2$) contains two non-adjacent vertices, say, $u,u'$ (respectively $v,v'$).
But then $\{u,u',v,v'\}$ induces a 4-hole. So at least one of $A_1$ and  $A_2$ induces
a clique. By symmetry, we assume that $A_1$ induces a clique.

If $V_1 \neq \emptyset$, then $A_1 \cup K$ is a clique cutset separating
$V_1$ from $V_2 \cup A_2$. So $V_1 = \emptyset$, and therefore $|A_1| \geq 2$.
But then any two vertices of $A_1$ are twins in $G$.
\end{proof}

Note that the proof of \autoref{lem:twin} makes use of merely the absence of $4$-holes.
The following decomposition of (even-hole,cap)-free graphs is an immediate consequence
of \autoref{thm:decom cap} and \autoref{lem:twin}.
\begin{thm}\label{thm:ehcap}
Suppose that $G$ is (even-hole,cap)-free graph that contains no universal vertices,
no twin vertices, and no clique cusets. Then $G$ is triangle-free.
\end{thm}


\section{Coloring (even-hole, cap)-free graphs}\label{sec:approx algo}

In this section, we prove our main result in this paper.
Then we turn our proof into a polynomial-time approximation algorithm.
First we note that `two vertices being twin vertices' in fact
defines an equivalence relation.
\begin{obse}
Let $\sim_T$ be the binary relation of two vertices being twins.
Then $\sim_T$ is an equivalence relation. Moreover,
each equivalence class is a clique and for any two equivalence classes $X$ and $Y$,
$X$ and $Y$ are either complete or anti-complete to each other.
 \end{obse}

 \begin{proof}
 Clearly, $\sim_T$ is reflexive and symmetric.
 It remains to show transitivity. Suppose that
 $u$, $v$ and $w$ are three vertices so that
 $u$, $v$ are twins and $v$, $w$ are twins.
 Then $N[u]=N[v]=N[w]$. Therefore, $u$ and $w$ are twins.
Since any pair of twins are adjacent, each equivalence class is a clique.
Similarly, if a vertex in class $X$ is adjacent to a vertex in class $Y$,
then $X$ is complete to $Y$; otherwise $X$ is anti-complete to $Y$.
 \end{proof}


%

We are now ready to prove \autoref{thm:ehbf_main}.

\begin{proof}[Proof of \autoref{thm:ehbf_main}]
We prove the theorem by induction on $|G|$.
We may assume that $G$ is connected, for otherwise
applying inductive hypothesis to each connected component of $G$ completes the proof.
If $G$ contains a universal vertex $u$, then $G-u$
has $\chi(G-u)\le \frac{3}{2}\omega(G-u)$.
Clearly, $\chi(G)=\chi(G-u)+1$ and $\omega(G)=\omega(G-u)+1$.
It follows that
\[\chi(G)=\chi(G-u)+1\le \frac{3}{2}\omega(G-u)+1=\frac{3}{2}(\omega(G)-1)+1\le \frac{3}{2}\omega(G).\]

If $G$ contains a clique cutset $K$, then $G-K$ is a disjoint union of two subgraphs $H_1$ and $H_2$.
Let $G_i=H_i\cup K$ for $i=1,2$. Then $\chi(G)=\max\{\chi(G_1),\chi(G_2)\}$.
Thus,
\[\chi(G)=\max\{\chi(G_1),\chi(G_2)\}\le \max\{\frac{3}{2}\omega(G_1),\frac{3}{2}\omega(G_2)\}\le \frac{3}{2}\omega(G).\]

Therefore, $G$ has no universal vertices or clique cutsets. Now we partition $V(G)$
into equivalence classes $T_1,T_2,\ldots,T_r$ under $\sim_T$.
Take an arbitrary vertex $t_i\in T$ for $1\le i\le r$ and let $G'=G[\{t_1,\ldots,t_r\}]$.
Note that $G'$ is obtained from $G$ by successively removing twin vertices.
We claim that removing twin vertices does not create a clique cuset or a universal vertex.

\begin{cla}\label{lem:twin removal}
Suppose that $u$ and $v$ are twin in $G$. Then $G-u$ does not contain
any universal vertex or clique cutset.
 \end{cla}

\begin{proof}[Proof of \autoref{lem:twin removal}]
Suppose not.
If $G-u$ contains a universal vertex, say $x$.
Then $x$ is adjacent to each vertex in $G-u$, in particular to $v$.
This implies that $x$ is also adjacent to $u $, since $u$ and $v$ are twins in $G$.
Now $x$ is a universal vertex in $G$, contradicting to our assumption.
So, $G-u$ contains no universal vertices.

Suppose that $G-u$ contains a clique cutset $K$.
Now $G-K$ is the disjoint union of two vertex-disjoint subgraphs $H_1$ and $H_2$.
Let $G_i$ be the subgraph of $G$ induced by $V(H_i)\cup K$ for $i=1,2$.
If $v$ is in $H_1$ or $H_2$, then $K$ is still a clique cutset in $G$.
So, $v\in K$. But then $K\cup \{u\}$ is a clique cutset of $G$, a contradiction.
\end{proof}

By \autoref{lem:twin removal} and \autoref{thm:ehcap}
we conclude that $G'$ is triangle-free, and so $\chi(G')\le 3$ by \autoref{coro:omegaplus1}.
On the other hand, note that $G'$ is connected (since $G$ is connected).
In particular, each vertex of $G'$ lies in an edge of $G'$.
Therefore, any maximal clique in $G'$ is an edge.
This means that any maximal clique of $G$ is a union of two $T_j$'s,
which implies that $\omega(G-G')= \omega(G)-2$.
By inductive hypothesis, $\chi(G-G')\le \frac{3}{2}\omega(G-G')$.
Then
\[\chi(G)\le \chi(G-G')+\chi(G')\le \frac{3}{2}\omega(G-G')+3= \frac{3}{2}(\omega(G)-2)+3=
\frac{3}{2}\omega(G).\]
Since $\chi(G)$ is an integer, the theorem follows.
\end{proof}

\begin{figure}[htbp]
\centering
\includegraphics[width=0.3\textwidth]{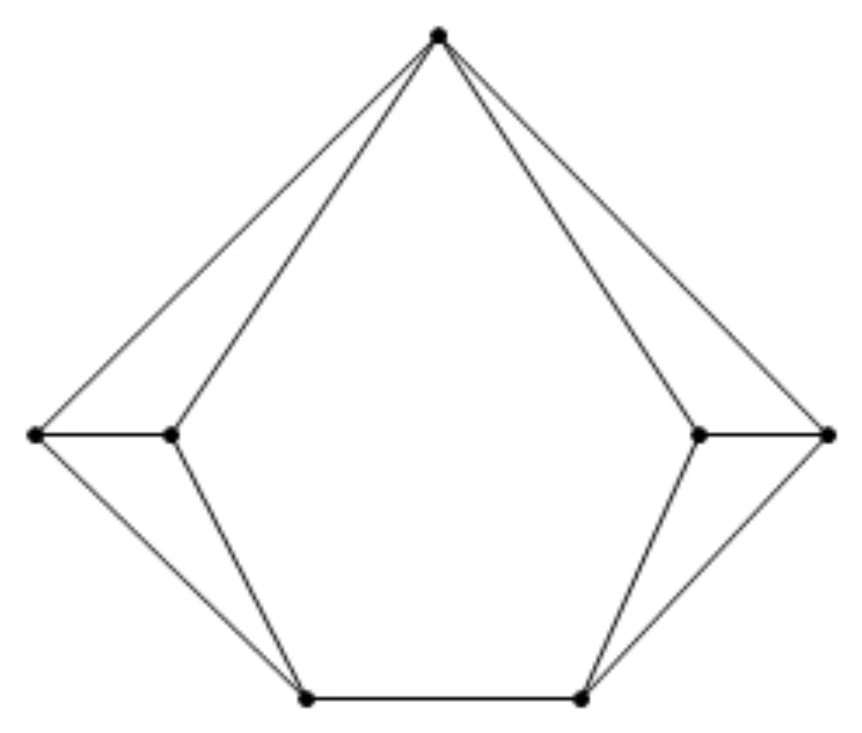}
\caption{The Hajos graph.}
\hypertarget{hajos}{}
\label{hajos}
\end{figure}

The bound in \autoref{thm:ehbf_main} is attained by odd-holes
and the \hajos.
Note that these graphs have clique number at most $3$. For graphs
with large clique number, we do not have an example
showing that the bound is tight. Nevertheless, the optimal constant
is at least $5/4$. For any integer $k\ge 1$, let $G_k$ be the graph obtained from a $5$-hole by  replacing
each vertex of the $5$-hole with a clique of size $2k$ and making two cliques complete
(respectively anti-complete) if the two original vertices are adjacent (respectively non-adjacent)
on the $5$-hole. Clearly, $|G_k|=10k$, $\alpha(G_k)=2$ and $\omega(G_k)=4k$.
Hence, $\chi(G_k)\ge \frac{|G_k|}{\alpha(G_k)}=5k$.
Moreover, it is easy to see that $G_k$ does admit a $5k$-coloring.
So,  $\chi(G_k)=5k=\frac{5}{4}\omega(G_k)$.
A natural question is that whether or not  one can reduce $\lfloor3/2\rfloor$ to $\lceil5/4\rceil$.
\begin{prob}
Is it true that $\chi(G)\le \lceil\frac{5}{4}\omega(G)\rceil$ for every (even-hole,cap)-free graph $G$?
\end{prob}

It was shown in \bfcite{CKS07} that
this is true for the class of $(C_4,P_5)$-free graphs which is a subclass of (even-hole,cap)-free graphs.

\textbf{A $3/2$- approximation algorithm}

We now turn our proof of \autoref{thm:ehbf_main}  into a $3/2$-approximation algorithm
for computing $\chi(G)$ if $G$ is (even-hole,cap)-free.
The algorithm outputs a $\frac{3}{2}\omega(G)$-coloring
of $G$ in polynomial time. We need one more observation.

\begin{obse}\label{obse:ucs}
Suppose that  $G$ is a graph without clique cutsets.
If  $u\in V(G)$ is a universal vertex,
then $G-u$ contains no clique cutsets.
\end{obse}

\begin{proof}
If $K$ is a clique cutset in $G-u$, then $K\cup \{u\}$
is a clique cutset in $G$.
\end{proof}

The proof of \autoref{thm:ehbf_main} is almost algorithmic except for the last step
where we deal with $G$ with no clique cutsets or universal vertices.
Essentially we want  to successively remove
a triangle-free subgraph, one vertex from each equivalence class,  from $G$
so that the removal of it reduces the clique number of the graph exactly by $2$.
During the removal process, however, if the graph becomes disconnected, a maximal clique
could just be one equivalence class, say $T_i$ (that forms a connected component of the graph).
If  $T_i$ happens to be a maximum clique of the current graph, then
removing a single vertex from $T_i$ may reduce the clique number by at most $1$.
This happens when either $T_i$ has at least two vertices or the current graph is just
an independent set. But both cases have an easy fix. In the former case,
we simply remove two vertices from $T_i$, and in the latter case
we color the independent set with one new color that has not been used
(at this point all vertices of $G$ have been colored).
Clearly, the number of subgraphs we removed is at most $\omega(G)/2=O(n)$.
Moreover, each time it takes $O(m+n)$ time (determining the connected components)
to find such a subgraph. Therefore, it takes $O(mn)$ time in total for finding subgraphs.
On the other hand, by \autoref{coro:coldiamond} we can color all subgraphs
in $O(n^2)$ time.
\begin{lem}\label{lem:colatom}
Suppose that $G$ is a (even-hole,cap)-free graph without universal vertices or clique cutsets.
If the equivalence classes $T_i$'s under $\sim_T$ are given,
one can find a $\frac{3}{2}\omega(G)$-coloring for $G$
in $O(mn)$ time.
\end{lem}

We now present the algorithm for general (even-hole,cap)-free graphs.

\begin{algorithm}[H]
\SetAlgoLined
\KwIn{A (even-hole,cap)-free graph $G$.}
\KwOut{A $\frac{3}{2}\omega(G)$-coloring of $G$.}
Do clique cutset decomposition of $G$ \bfcite{Ta85} and obtain a binary decomposition tree $T(G)$.

\For{each atom $A$}{
	$A':=A$\;
	\For{each $a\in V(A)$}{
		\If(\tcp*[h]{$a$ is a universal vertex in $A$}){$|N_{A}(a)|=|A|-1$}
			{$A':=A-a$\;}
	}
	Partition $A'$ into equivalence classes $T_1,\ldots,T_r$ under $\sim_T$\; 

Obtain a $\frac{3}{2}\omega(A')$-coloring $\phi_{A'}$ of $A'$ by \autoref{lem:colatom}\;
Extend $\phi_{A'}$ to a coloring $\phi_A$ of $A$ by coloring each vertex in $A\setminus A'$ with a  new color\;
}

Combine coloring $\phi_A$ of the atoms along $T(G)$ and obtain a coloring $\phi$ of $G$.
\caption{A $3/2$-approximation algorithm for \textsc{Chromatic Number}}
\label{alg:1.5-approx}
\end{algorithm}

%

We show that the algorithm is correct.
\begin{thm}
\autoref{alg:1.5-approx} is correct and runs in $O(mn^2)$ time.
\end{thm}

\begin{proof}
We first discuss the running time.
The clique cutset decomposition can be found in $O(mn)$
time and there are at most $n$ atoms, see \bfcite{Ta85}.
The \textbf{for} loop from line $4$  to line $8$ and line $11$ apparently take $O(n)$ time.
To partition $A'$ into $T_1,\ldots,T_r$, we test for each edge $e=xy\in E(A')$
whether or not $N[x]=N[y]$. For each edge it takes $O(n)$ time and
so line $9$ takes $O(mn)$ time. Line $10$ takes $O(mn)$ time by \autoref{lem:colatom}.
In a word, the coloring $\phi_{A}$, for each atom $A$,
can be found inn $O(mn)$ time. Since there are $O(n)$ atoms, the total running time is $O(mn^2)$.

To prove the correctness, we first note that  $A'$ (at the end of line $8$) contains no universal vertex.
Suppose not, let $b\in A'$ be a universal vertex in $A'$.
Since all vertices $A\setminus A'$ are universal vertices in $A$, they are all adjacent to $b$.
This implies that $b$ is a universal vertex in $A$ and so it would have been removed
during the \textbf{for} loop from line $4$  to line $8$, a contradiction.
Furthermore, $A'$ contain no clique cutsets by \autoref{obse:ucs}.
Therefore, the correctness follows from \autoref{lem:colatom} and the fact that
universal vertices and clique cutsets preserve the $\chi$-binding function.
\end{proof}

\bibliographystyle{plain} \bibliography{references}

\end{document}